\documentclass[leqno]{amsart}

\usepackage[round,comma,numbers]{natbib}                

\usepackage{amsmath, amsthm}
\usepackage{amsfonts,amssymb,dsfont}
\usepackage{nicefrac,mathrsfs}
\usepackage{bm}                                 
\usepackage[backgroundcolor=white,bordercolor=orange]{todonotes}
\usepackage{bbm}

\newtheorem{theorem}{Theorem}[section]
\newtheorem{lem}[theorem]{Lemma}              
\newtheorem{pro}[theorem]{Proposition}  
\newtheorem{rem}[theorem]{Remark} 

\theoremstyle{definition}

\newtheorem{dfn}{Definition}[section]

\setcitestyle{square}
\usepackage{eurosym}
\DeclareRobustCommand{\officialeuro}{%
	\ifmmode\expandafter\text\fi
	{\fontencoding{U}\fontfamily{eurosym}\selectfont e}}

\usepackage[pdftex,colorlinks,urlcolor=red,citecolor=blue,linkcolor=red]{hyperref}
\setcitestyle{square}
\usepackage{comment}

\begin{document}

\title[Hyperfinite Construction of $G$-expectation]{Hyperfinite Construction of $G$-expectation}

\author{Tolulope Fadina}
\address{Faculty of Mathematics, Bielefeld University, D-33615 Bielefeld, Germany. Email: tfadina@math.uni-bielefeld.de, tolulope.fadina@stochastik.uni-freiburg.de}
\author {Frederik Herzberg}
    \address{Center for Mathematical Economics (IMW), Bielefeld University, D-33615 Bielefeld, Germany. Email: fherzberg@uni-bielefeld.de}
\thanks{We are very grateful to Patrick Beissner, Yan Dolinsky, and Frank Riedel for helpful comments and suggestions. This work was supported by the International Graduate College (IGK) {\em Stochastics and Real World Models} (Bielefeld--Beijing) and the Rectorate of Bielefeld University (Bielefeld Young Researchers' Fund)}

\maketitle

\begin{abstract}
\noindent 
The {\em hyperfinite $G$-expectation} is a nonstandard discrete analogue of $G$-expectation (in the sense of Robinsonian nonstandard analysis). A {\em lifting} of a continuous-time $G$-expectation operator is defined as a hyperfinite $G$-expectation which is infinitely close, in the sense of nonstandard topology, to the continuous-time $G$-expectation. We develop the basic theory for hyperfinite $G$-expectations and prove an existence theorem for liftings of (continuous-time) $G$-expectation. For the proof of the lifting theorem, we use a new discretization theorem for the $G$-expectation (also established in this paper, based on the work of Dolinsky, Nutz and Soner [Stoch. Proc. Appl. 122, (2012), 664--675]). 


\end{abstract}

\newcommand{\ulambda}{\underline{\lambda}}
\newcommand{\olambda}{\overline{\lambda}}
\newtheorem{thm}{Theorem}

\newtheorem{remark}[thm]{Remark}

\newtheorem{assumption}{Assumption}
\newtheorem{case}{Case}
\newtheorem{nota}[thm]{Notational convention}
\newcommand {\PP}{\mathbb P}

\newcommand {\QQ}{\mathbb Q}
\newcommand {\half}{\frac{1}{2}}

\newcommand {\fP}{\mathfrak{P}}
\newcommand\norm[1]{\left\lVert#1\right\rVert}

\section{Introduction}
\citet{Dolinsky} showed a Donsker-type result for $G$-Brownian motion
by introducing a notion of volatility uncertainty in discrete time and defined a discrete version of \textit{Peng's} $G$-expectation. 
In the continuous-time limit, the resulting sublinear expectation converges weakly to $G$-expectation. In their discretization, \citet{Dolinsky} allow for martingale laws whose support is the whole set of reals in a $d$-dimensional setting. In other words, they only discretize the time line, but not the state space of the canonical process. Now for certain applications, for example, a hyperfinite construction of $G$-expectation in the sense of Robinsonian nonstandard analysis, a discretization of the state space would be necessary. 
Thus, we develop a modification of the construction by \citet{Dolinsky} which even ensures that the sublinear expectation operator for the discrete-time canonical process corresponding to this discretization of the state space 
(whence the martingale laws are supported by a finite lattice only) converges to the $G$-expectation. Further, we prove a lifting theorem, in the sense of Robinsonian nonstandard analysis, for the $G$-expectation. Herein, we use the discretization result for the $G$-expectation.

Nonstandard analysis makes consistent use of infinitesimals in mathematical analysis based on techniques from mathematical logic. This approach is very promising because it also allows, for instance, to study continuous-time stochastic
processes as formally finite objects. Many authors have applied nonstandard analysis to problems in measure theory, probability theory and mathematical economics (see for example, \citet{Raimondo} and the references therein or the 
contribution in \citet{Imme}), especially after \citet{Loeb} converted nonstandard measures (i.e. the images of standard measures under the nonstandard embedding ${}^{*}$) into real-valued, 
countably additive measures, by means of the standard part operator and \textit{Caratheodory}'s extension theorem. 
One of the main ideas behind these applications is the extension of the notion of a finite set known as {\em hyperfinite set} or more causally, 
a formally finite set. Very roughly speaking, hyperfinite sets are sets that can be formally enumerated with both standard and nonstandard natural numbers up to a (standard or nonstandard, i.e. unlimited) natural number. 

\citet{Anderson, Keisler, Toml, Edwin}, a few to mention, used Loeb's \cite{Loeb} approach to develop basic nonstandard stochastic analysis and in particular, the nonstandard It{\^{o}} calculus. \citet{Loeb} also presents the construction of a Poisson processes using nonstandard analysis. 
\citet{Anderson} showed that Brownian motion can be constructed from a hyperfinite number of coin tosses, and provides a detailed proof using a special case of Donsker's theorem. \citet{Anderson} also gave a nonstandard construction of stochastic
integration with respect to his construction of Brownian motion. \citet{Keisler} uses Anderson's \cite{Anderson} result to obtain some results on stochastic differential equations. \citet{Toml} gave the hyperfinite construction ({\em lifting}) of $L^2$ standard martingales. 
Using nonstandard stochastic analysis, \citet{Perkins81} proved a global characterization of (standard) Brownian local time.
In this paper, we do not work on the Loeb space because the $G$-expectation and its corresponding $G$-Brownian motion are not based on a classical probability measure, but on a set of martingale laws.

The aim of this paper is to give two approximation results on $G$-expectation. First, to refine the discretization of $G$-expectation by \citet{Dolinsky}, in order to obtain a discretization of  the sublinear expectation where the martingale laws are defined on a finite lattice rather than the whole set of reals. Second, to give an alternative, combinatorially inspired construction of the $G$-expectation based on the discretization result. We hope that this result may eventually become useful for applications in financial economics (especially existence of equilibrium on continuous-time financial markets with volatility uncertainty) and provides additional intuition for {\em Peng}'s $G$-stochastic calculus. 
We begin the nonstandard treatment of the $G$-expectation by defining a notion of $S$-continuity, a standard part operator, and proving a corresponding lifting (and pushing down) theorem. Thereby, we show that our hyperfinite construction is the appropriate nonstandard analogue of the $G$-expectation.

The rest of this paper is divided into two parts: in the first part, Section \ref{weakapprox}, we define Peng's $G$-expectation and introduce a discrete-time analogue of a $G$-expectation in the spirit of \citet{Dolinsky}. Unlike in \citet{Dolinsky}, we require the discretization of the martingale laws to be defined on a finite lattice rather than the whole set of reals. In the continuous-time limit, the resulting sublinear expectation converges weakly to the continuous-time $G$-expectation.   
In the second part, Section \ref{hyperfinitecons},   we develop the basic theory for hyperfinite $G$-expectations and prove an existence theorem for liftings of (continuous-time) $G$-expectation. We extend the discrete time analogue of the $G$-expectation in Section \ref{weakapprox} to a hyperfinite time analogue. Then, we use the characterization of convergence in nonstandard analysis to prove that the hyperfinite discrete-time analogue of the $G$-expectation is infinitely close in the sense of nonstandard topology to the continuous-time $G$-expectation.

\section{Weak approximation of $G$-expectation with discrete state space}
\label{weakapprox}

\citet{Peng} introduced a sublinear expectation on a well-defined space $\mathbb{L}^1_{G}$, the completion of $\text{Lip}_{b.cyl}(\Omega)$ (bounded and Lipschitz cylinder function) under the norm $\| \cdot \|_{\mathbb{L}^1_{G}}$, 
under which the increments of the canonical process $(B_t)_{t>0}$ are zero-mean, independent and stationary and can be proved to be $(G)$-normally distributed. This type of process is called {\em G-Brownian motion}
and the corresponding sublinear expectation is called {\em G-expectation}. 

The $G$-expectation $\xi \mapsto \mathcal{E}^{G}(\xi)$ is a sublinear operator defined on a class of random variables on $\Omega$. The symbol $G$ refers to a given function 
\begin{equation}
\label{gfunc}
G(\gamma):= \frac{1}{2} \sup_{c \in \mathbf{D}} c \gamma :\mathbb{R} \rightarrow \mathbb{R}
\end{equation}
where $\mathbf{D}=[r_{\mathbf{D}},R_{\mathbf{D}}]$ is a nonempty, compact and convex set, and $0\leq r_{\mathbf{D}} \leq R_{\mathbf{D}} < \infty$ are fixed numbers. The construction of the $G$-expectation is as follows. Let $\xi = f(B_T)$, 
where $B_T$ is the $G$-Brownian motion and $f$ a sufficiently regular function. Then $\mathcal{E}^{G}(\xi)$ is defined to be the initial value $u(0,0)$ of the solution of the nonlinear backward heat
equation, 
$$-\partial_t u - G(\partial^2_{xx} u) = 0,$$
with terminal condition $u(\cdot, T)=f$, \citet{Pardoux}. The mapping $\mathcal{E}^{G}$ can be extended to random variables of the form ${\xi = f(B_{t_1}, \cdots, B_{t_n} )}$ by a stepwise evaluation of the PDE and then to the completion $\mathbb{L}^1_{G}$ of the space of all such random variables (cf. \citet{Dolinsky}). 
\citet{Denis} showed that $\mathbb{L}^1_{G}$ is the completion of $\mathcal{C}_b(\Omega)$ and $\text{Lip}_{b.cyl}(\Omega)$ under the norm ${\| \cdot \| _{\mathbb{L}^1_{G}}}$, and that $\mathbb{L}^1_{G}$ is the space of the so-called quasi-continuous function and contains all bounded continuous functions on the canonical space $\Omega$, but not all bounded measurable functions are included. \citet{Ruan} introduced the invariance principle of $G$-Brownian motion using the theory of sublinear expectation.
There also exists an equivalent alternative representation of the $G$-expectation known as the {\em dual view on $G$-expectation via volatility uncertainty}, see \citet{Denis}:
\begin{equation}
\label{superhedge}
\mathcal{E}^{G}(\xi) = \sup _{P \in \mathcal{P}^G} \mathbb{E}^{P} [\xi], \quad \xi= f(B_T),
\end{equation}
where $\mathcal{P}^G$ is defined as the set of probability measures on $\Omega$ such that, for any $P \in \mathcal{P}^G$, $B$ is a martingale with 
the volatility $d\left\langle B \right\rangle _t / dt \in \mathbf{D}$ $P \otimes dt$ a.e.

\subsection{Continuous-time construction of sublinear expectation} 
\label{G-framework}
Let $\Omega = \{ \omega \in \mathcal{C}([0,T]; \mathbb{R}) : \omega_0=0 \}$ be the canonical space endowed with the uniform norm ${\|\omega\| _{\infty} = \sup _{0 \leq t \leq T} |\omega _t| },$ where $|\cdot|$ 
denotes the absolute value
on $\mathbb{R}$. Let $B$ be the canonical process $B_t(\omega) = \omega_t,$ and $\mathcal{F}_t = \sigma (B_s, 0 \leq s \leq t)$ 
the filtration generated by $B$. 
A probability measure $P$ on $\Omega$ is a martingale law provided $B$ is a $P$-martingale and $B_0 = 0$ $P$ a.s. Then, $\mathcal{P}_{\mathbf{D}}$ is the set of martingale laws on $\Omega$ and 
the volatility takes values in $\mathbf{D}$, $P \otimes dt$ a.e; 
\begin{equation*}
\mathcal{P}_{\mathbf{D}} = \left\{ P \text{ martingale law on } \Omega \text{: } d\left\langle B \right\rangle _t / dt \in \mathbf{D} \text{, } P \otimes dt \text{ a.e.} \right\}.
\end{equation*}

\subsection{Discrete-time construction of sublinear expectation}We denote
\begin{equation*}
\mathcal{L}_n = \left\{ \frac{j}{n \sqrt{n}}, \quad -n^2 \sqrt{R_{\mathbf{D}}} \leq j \leq  n^2 \sqrt{R_{\mathbf{D}}}, \quad \text{for } j\in \mathbb{Z}\right\},
\end{equation*}
and $\mathcal{L}^{n+1}_{n}= \mathcal{L}_n \times \dots \times \mathcal{L}_n  (n+1 \text{  times})$, for $n\in \mathbb{N}$. 
Let $X^n = {(X^n_k)}_{k=0}^{n}$ be the canonical process $X^n_k(x)= x_k$ defined on $\mathcal{L}^{n+1}_{n}$ and ${(\mathcal{F}_k^n)}_{k=0}^{n} = \sigma (X_l^n, l =0, \dots, k)$ be the filtration generated by $X^n$. We note that ${R_{\mathbf{D}} = \sup_{\alpha \in \mathbf{D}}|\alpha|}$.
$$\mathbf{D}^{\prime}_n = \mathbf{D} \cap \left( \frac{1}{n} \mathbb{N} \right)^2$$
is a nonempty bounded set of volatilities. A probability measure $P$ on $\mathcal{L}_{n}^{n+1}$ is a martingale law provided $X^n$ is a $P$-martingale and $X^n_0=0$ $P$ a.s. The increment $\Delta X^n_k = X^n_k - X^n_{k-1}$.
Let $\mathcal{P}_{\mathbf{D}}^{n}$ be the set of martingale laws of $X^n$ on $\mathbb{R}^{n+1} $, i.e.,
\begin{equation*}
\mathcal{P}_{\mathbf{D}}^{n} = \left\{ P \text{ martingale law on } \mathbb{R}^{n+1}  \text{: } r_{\mathbf{D}} \leq |\Delta X^n_k|^2 \leq R_{\mathbf{D}} \text{, }P \text{ a.s.} \right\},
\end{equation*} 
such that for all $n$, $ \mathcal{L}_{n}^{n+1} \subseteq \mathbb{R}^{n+1} $. 
\subparagraph*{ }In order to establish a relation between the continuous-time and discrete-time settings, we obtained a continuous-time process $\widehat{x} _t \in \Omega$ from any discrete path $x \in \mathcal{L}^{n+1}_{n}$ by linear interpolation. i.e.,
$$\widehat{x}_t := (\lfloor nt/T \rfloor +1 -nt/T)x_{ \lfloor nt/T \rfloor } + (nt/T - \lfloor nt/T \rfloor)x_{\lfloor nt/T \rfloor +1}$$
where $\widehat{   }: \mathcal{L}^{n+1}_{n} \rightarrow \Omega $ is the linear interpolation operator, $x = (x_0, \dots, x_n) \mapsto \widehat{x} = \{(\widehat{x})_{0\leq t \leq T}\} ,$ and $ \lfloor y \rfloor $ denotes the greatest integer less than or equal to $y.$ If $X^n$ is the canonical process on $\mathcal{L}^{n+1}_{n}$ and $\xi$ 
is a random variable on $\Omega,$ then $\xi (\widehat{X}^n)$ defines a random variable on $\mathcal{L}^{n+1}_{n}.$

\subsection{Strong formulation of volatility uncertainty}
We consider martingale laws generated by stochastic integrals with respect to a fixed Brownian motion as in \citet{Dolinsky,Marcel} and a fixed random walk as in \citet{Dolinsky}.\\
Continuous-time construction; let $\mathcal{Q}_{\mathbf{D}} $ be the set of martingale laws:
\begin{equation*}
\mathcal{Q}_{\mathbf{D}} = \left\{ P _{0} \circ (M)^{-1}; \text{ } M = \int f (t, B) dB_t, \text{ and }  f \in \mathcal{C} ([0,T] \times \Omega ; \sqrt{\mathbf{D}}) \text{ is adapted} \right\}.
\end{equation*}
$B$ is the canonical process under the Wiener measure $P_0 $.\\
Discrete-time construction; we fix $n \in \mathbb{N}$, $\Omega _n = \lbrace \omega = (\omega _1, \dots, \omega _n) : \omega _i \in \lbrace \pm 1 \rbrace, \quad i =1, \dots,n \rbrace$ equipped with the power set and let 
$$P_n = \underbrace{ \frac{\delta_{-1}+ \delta_{+1}}{2} \otimes \cdots \otimes  \frac{\delta_{-1}+ \delta_{+1}}{2} }_\text{n times}$$ 
be the product probability associated with the uniform distribution where $\delta_{x}(A)$ is a Dirac 
measure for any $A \subseteq \mathbb{R}$ and a given $x \in A$.   
Let $\xi _1, \dots, \xi_n$ be an i.i.d sequence of $\lbrace \pm 1 \rbrace$-valued random variables. The components of $\xi _k$ are orthonormal in $L^2(P_n)$ and the associated scaled random walk is $$ \mathbb{X} = \frac{1}{\sqrt{n}} \sum _{l=1}^k \xi _l. $$
We denote by $\mathcal{Q}_{\mathbf{D}^{\prime}_n}^{n}$ the set of martingale laws of the form:
\begin{align}
\label{disc-meas} 
\mathcal{Q}_{\mathbf{D}^{\prime}_n}^{n} = \left\{  P_n  \circ (M^{f, \mathbb{X}})^{-1}; \text{  } f: \lbrace 0, \dots, n \rbrace \times \mathcal{L}^{n+1}_{n} \rightarrow \sqrt{\mathbf{D}^{\prime}_n} \text{ is } \mathcal{F}^n \text{-adapted.} 
\right\}
\end{align}
where $ M^{f, \mathbb{X}} = \left( \sum _{l=1}^{k} f (l-1,\mathbb{X})\Delta \mathbb{X}_ l  \right)_{k=0}^{n}.$

\subsection{Results and proofs}
\label{G-result}
Theorem \ref{max-result}  states that a sublinear expectation with discrete-time volatility uncertainty on our finite lattice converges to the $G$-expectation.

\begin{lem}
	\label{equivmeas}

	$\mathcal{Q}_{\mathbf{D}}^{n} = \left\{ P_n  \circ \left( M^{f,\mathbb{X}} \right) ^{-1}; \text{ }f: \lbrace 0, \dots, n \rbrace \times \mathbb{R}^{n+1} \rightarrow \sqrt{\mathbf{D}} \text{ is adapted}\right\}.
	$ Then $\mathcal{Q}_{\mathbf{D}}^{n} \subseteq \mathcal{P}_{\mathbf{D}}^{n}$. 
\end{lem}

\begin{pro}
	\label{mthm}

	Let $\xi: \Omega \rightarrow \mathbb{R}$ be a continuous function satisfying $|\xi(\omega)| \leq a(1+\parallel \omega\parallel_{\infty})^b$ for some constants $a,b >0.$ Then, 
	\begin{itemize}
		\item[$(i)$]
		\begin{equation}
		\label{main-thm}
		\lim _{n \rightarrow \infty} \sup_{\mathbb{Q} \in \mathcal{Q}_{\mathbf{D}^{\prime}_n/n}^{n}} \mathbb{E}^{\mathbb{Q}} [\xi (\widehat{X}^n)]  = \sup_{P \in \mathcal{Q}_{\mathbf{D}}} \mathbb{E}^{P} [\xi].
		\end{equation} 
		\item[$(ii)$]
		\begin{equation}
		\label{max-eq}
		\sup_{\mathbb{Q} \in \mathcal{Q}_{\mathbf{D}^{\prime}_n/n}^{n}} \mathbb{E}^{\mathbb{Q}} [\xi (\widehat{X}^n)] =  \max_{\mathbb{Q} \in \mathcal{Q}_{\mathbf{D}^{\prime}_n/n}^{n}} \mathbb{E}^{\mathbb{Q}} [\xi (\widehat{X}^n)].
		\end{equation}
	\end{itemize}
\end{pro}

To prove \eqref{main-thm}, we prove two separate inequalities together with a density argument.
The left-hand side of \eqref{max-eq} can be written as
\begin{equation*}
\sup_{\mathbb{Q} \in \mathcal{Q}_{\mathbf{D}^{\prime}_n/n}^{n}} \mathbb{E}^{\mathbb{Q}} [\xi (\widehat{X}^n)] =  \sup_{f \in \mathcal{A} } \mathbb{E}^{P_n \circ (M^{f, \mathbb{X}})^{-1}} [\xi (\widehat{X}^n)],
\end{equation*}
where $\mathcal{A} =  \left\lbrace f: \lbrace 0, \dots, n \rbrace \times \mathcal{L}^{n+1}_{n} \rightarrow \sqrt{\mathbf{D}^{\prime}_n/n}\text{ is } \mathcal{F}^n \text{-adapted.}\right\rbrace$. We prove that $\mathcal{A}$ is a compact subset of a finite-dimensional 
vector space, and that $f\mapsto \mathbb{E}^{P_n \circ (M^{f, \mathbb{X}})^{-1}} [\xi (\widehat{X}^n)]$ is continuous. Before then, we introduce a smaller space $\mathbb{L}_{*}^1$ that is defined as the completion of $\mathcal{C}_b (\Omega; \mathbb{R})$ under the norm
{(cf. \citet{Dolinsky})}
\begin{equation*}
\parallel \xi \parallel _{*} := \sup _{Q \in \mathcal{Q}} \mathbb{E} ^{Q} |\xi| , \quad \mathcal{Q}:=  \mathcal{P}_{\mathbf{D}} \cup \lbrace P \circ (\widehat{X}^n)^{-1};  P \in \mathcal{P}_{\mathbf{D}/n}^{n}, { }  n \in \mathbb{N}. \rbrace .  
\end{equation*}
This is because Proposition \ref{mthm} will not hold if $\xi$ just belong to $\mathbb{L}^1_{G}$, which is the completion of $\mathcal{C}_b (\Omega; \mathbb{R})$ under the norm 
\begin{equation}
\label{dnorm}
\parallel \xi \parallel _{\mathbb{L}^1_{G}} := \sup _{P \in \mathcal{P}_{\mathbf{D}}} \mathbb{E} ^{P} [|\xi|].  
\end{equation}

\begin{proof}[Proof of Proposition \ref{mthm}]
	\textit{First inequality \textnormal{(for $\leq$ in \eqref{main-thm})}:}
	\begin{equation}
	\label{first-inequality}
	\limsup _{n \rightarrow \infty} \sup_{\mathbb{Q} \in \mathcal{Q}_{\mathbf{D}^{\prime}_n/n}^{n}} \mathbb{E}^{\mathbb{Q}} [\xi (\widehat{X}^n)]  \leq  \sup_{P \in \mathcal{Q}_{\mathbf{D}}} \mathbb{E}^{P} [\xi].
	\end{equation}
	For all $n$, $\sqrt{\mathbf{D}^{\prime}_n/n} \subseteq \sqrt{\mathbf{D}/n}$ and $\mathcal{Q}_{\mathbf{D}^{\prime}_n}^{n} \subseteq \mathcal{Q}_{\mathbf{D}}^{n}$.
	It is shown in \citet{Dolinsky} that 
	\begin{equation*}
	\limsup_{n\rightarrow \infty} \sup_{\mathbb{Q} \in \mathcal{P}_{\mathbf{D}/n}^{n}} \mathbb{E}^{\mathbb{Q}} [\xi (\widehat{X}^n)] \leq \sup_{P \in \mathcal{P}_{\mathbf{D}}} \mathbb{E}^{P} [\xi].
	\end{equation*}
	Since $\mathcal{Q}_{\mathbf{D}} \subseteq \mathcal{P}_{\mathbf{D}}$ (see \citet[Remark~$3.6$]{Dolinsky}) 
	and $\mathcal{Q}^{n}_{\mathbf{D}} \subseteq \mathcal{P}^{n}_{\mathbf{D}}$ (see Lemma \ref{equivmeas}), \eqref{first-inequality} follows.
	
	\textit{Second inequality \textnormal{(for $\geq$ in \eqref{main-thm})}:} It remains to show that 
	$$\liminf _{n \rightarrow \infty} \sup_{\mathbb{Q} \in \mathcal{Q}_{\mathbf{D}^{\prime}_n/n}^{n}} \mathbb{E}^{\mathbb{Q}} [\xi (\widehat{X}^n)]  \geq  \sup_{P \in \mathcal{Q}_{\mathbf{D}}} \mathbb{E}^{P} [\xi].$$
	For arbitrary $P \in \mathcal{Q}_{\mathbf{D}}$, we construct a sequence ${(P^n)}_n$ such that for all  $n$,
	\begin{equation}
	\label{new-sq}
	P^n \in \mathcal{Q}_{\mathbf{D}^{\prime}_n/n}^{n},
	\end{equation}
	and 
	\begin{equation}
	\label{second}
	\mathbb{E}^{P}[\xi] \leq \liminf _{n\rightarrow \infty} \text{  } \mathbb{E}^{P^n}[\xi (\widehat{X}^n)].
	\end{equation}
	For fixed $n$, we want to construct martingales $M^n$ whose laws are in $\mathcal{Q}_{\mathbf{D}^{\prime}_n/n}^{n}$ and the laws of their interpolations tend to $P.$ 
	Thus, we introduce a scaled random walk with the piecewise constant c{\`a}dl{\`a}g property, 
	
	\begin{equation}
	\label{pr}
	W^n_t := \frac{1}{\sqrt{n}} \sum _{l=1}^{ \lfloor nt/T \rfloor } \xi _l = \frac{1}{\sqrt{n}} Z^n_{\lfloor nt/T \rfloor}, \quad 0\leq t \leq T,
	\end{equation}
	and we denote the continuous version of \eqref{pr} obtained by linear interpolation by 
	\begin{equation}
	\label{pir}
	\widehat{W}^n_t := \frac{1}{\sqrt{n}} \widehat{Z}^n_{\lfloor nt/T \rfloor}, \quad 0\leq t \leq T.
	\end{equation}
	By the central limit theorem; $(W^n, \widehat{W}^n) \Rightarrow (W,W) $ as $n \rightarrow \infty$ on $D([0,T];\mathbb{R}^2)$ ($\Rightarrow$ implies convergence in distribution). i.e., the law $(P_n) $ converges to the law $P_0$ on the Skorohod space $D([0,T];\mathbb{R}^2)$ 
	\citet[Theorem~$27.1$]{Patrick}. Let $g \in \mathcal{C}([0,T]\times \Omega, { } \sqrt{\mathbf{D}})$ such that
	$$P = P_0 \circ \left( \underbrace{ \int g(t, W) dW_t}_\text{$M$} \right)^{-1}.$$ 
	Since $g$ is continuous and $\widehat{W}^n_t$ is the interpolated version of \eqref{pr},
	$$\left( W^n, \left( g\left( \lfloor nt/T \rfloor T/n, \widehat{W}^n_t \right) \right)_{t \in [0,T]} \right) \Rightarrow \left(W, (g(t, W_t))_{t \in [0,T]} \right) \text{ as } n \rightarrow \infty \text{ on } D([0,T];\mathbb{R}^2).$$
	We introduce martingales with discrete-time integrals, 
	\begin{equation}
	\label{dis-mart}
	M^{n}_k := \sum _{l=1}^{k} g \left( (l-1)T/n, \widehat{W}^n \right) \widehat{W}^n_{lT/n} - \widehat{W}^n_{(l-1)T/n}.
	\end{equation}
	In order to construct $M^n$ which is ``close'' to $M$ and also is such that ${P_n \circ {(M^n)}^{-1}  \in \mathcal{Q}_{\mathbf{D}^{\prime}_n/n}^{n}}$. We choose ${\widetilde{h}_n: \{0, \cdots, n\} \times \Omega \rightarrow {\sqrt{{\mathbf{D}^\prime_n}/n}}} $ such that
	$$d_{J_1}\left( \left( \widetilde{h}_n (\lfloor nt/T \rfloor T/n, \widehat{W}^{n}_{t}) \right)_{t\in[0,T]} ,   \left( g (\lfloor nt/T \rfloor T/n, \widehat{W}^{n}_{t}) \right)_{t\in[0,T]}                    \right) $$
	is minimal (this is possible because there are only finitely many choices for $\left( \widetilde{h}_n (\lfloor nt/T \rfloor T/n, \widehat{W}^{n}_{t}) \right)_{t\in[0,T]} )$
	and $d_{J_1}$ is the Kolmogorov metric for the Skorohod $J_1$  topology. 
	From \citet[Theorem~$4.3$ and Definition~$4.1$]{Billingsley}, it follows that 
	$$ \left( W^n, \left( \widetilde{h}_n \left(\lfloor nt/T \rfloor T/n, \widehat{W}^{n}_{t} \right) \right)_{t\in[0,T]} \right) \Rightarrow \left(W, g(t, W_t)_{t \in [0,T]} \right) \text{ on }D([0,T];\mathbb{R}^2).$$
	We then define $g_n: \lbrace 0, \dots, n \rbrace \times \mathcal{L}^{n+1}_{n} \rightarrow {\sqrt{\mathbf{D}^{\prime}_n/n}}$ by 
	$g_n:( \ell, \vec{\mathbb{X}}) \mapsto \widetilde{h}_n ( \ell , \widehat{\vec{X}}).$
	Let $M^n$ be defined by 
	$$M^{n}_k = \sum _{l=1}^{k} g_n \left(l-1, \frac{1}{\sqrt{n}} Z^n \right) \frac{1}{\sqrt{n}} \Delta Z^n_l, \quad \forall k \in \{0,\cdots, n\}.$$
	By stability of stochastic integral (see \citet[Theorem~$4.3$ and Definition~$4.1$]{Duffie}), 
	$$\left( M^n_{\lfloor nt/T \rfloor} \right) _{t \in [0,T]} \Rightarrow M \quad \text{as } n \rightarrow \infty \text{ on }  D([0,T];\mathbb{R})$$
	because 
	$$M^n_{\lfloor nt/T \rfloor} = \sum_{l=1}^{\lfloor nt/T \rfloor} \widetilde{h}_n \left((l-1)T/n, \left(\widehat{W}_{kT/n} \right)_{k=0}^{n} \right) \Delta \widehat{W}_{lT/n}. $$
	In addition, as $n$ goes to $\infty$, the increments of $M^n$ uniformly tend to $0$. Thus, $\widehat{M}^n \Rightarrow M$ on $\Omega.$ Since $\xi$ is bounded and continuous,
	\begin{equation}
	\label{sk-conv}
	\lim_{n \rightarrow \infty} \mathbb{E}^{P_n \circ {(M^n)}^{-1}}[\xi(\widehat{X}^n)] = \mathbb{E}^{P_0 \circ {M}^{-1}}[\xi].
	\end{equation}
	Therefore, \eqref{new-sq} is satisfied for $P^n = P_n \circ {(M^n)}^{-1}  \in \mathcal{Q}_{\mathbf{D}^{\prime}_n/n}^{n}.$ Taking the $\liminf$ as $n$ tends to $\infty$ and the supremum over $P \in \mathcal{Q}_{\mathbf{D}},$ 
	\eqref{sk-conv} becomes
	\begin{equation}
	\label{second-inequality}
	\sup _{ P \in \mathcal{Q}_{\mathbf{D}}} \mathbb{E}^{P}[\xi] \leq \liminf _{n \rightarrow \infty} \sup _{ \mathbb{Q} \in \mathcal{Q}_{\mathbf{D}^{\prime}_n/n}^{n}} \mathbb{E}^{\mathbb{Q}}[\xi (\widehat{X}^n)].
	\end{equation}
	Combining \eqref{first-inequality} and \eqref{second-inequality}, 
	\begin{align*}
	{\sup_{P \in \mathcal{Q}_{\mathbf{D}}} \mathbb{E}^{P} [\xi]} & \geq {\limsup _{n \rightarrow \infty} \sup_{\mathbb{Q} \in \mathcal{Q}_{\mathbf{D}^{\prime}_n/n}^{n}} \mathbb{E}^{\mathbb{Q}} [\xi (\widehat{X}^n)]}\geq {\liminf _{n \rightarrow \infty} \sup_{\mathbb{Q} \in \mathcal{Q}_{\mathbf{D}^{\prime}_n/n}^{n}} \mathbb{E}^{\mathbb{Q}} [\xi (\widehat{X}^n)] } \geq {\sup_{P \in \mathcal{Q}_{\mathbf{D}}} \mathbb{E}^{P} [\xi]}.
	\end{align*}
	Therefore, 
	\begin{equation}
	\label{final-equation}
	\sup_{P \in \mathcal{Q}_{\mathbf{D}}} \mathbb{E}^{P} [\xi] = \lim _{n\rightarrow \infty} \sup_{\mathbb{Q} \in \mathcal{Q}_{\mathbf{D}^{\prime}_n/n}^{n}} \mathbb{E}^{\mathbb{Q}} [\xi (\widehat{X}^n)] .
	\end{equation}
	\textit{Density argument}: \eqref{main-thm} is established for all $\xi \in \mathcal{C}_b (\Omega, \mathbb{R})$.
	Since $\mathcal{Q}_{\mathbf{D}} \subseteq \mathcal{P}_{\mathbf{D}}$ (see \citet[Remark~$3.6$]{Dolinsky}) and $\mathcal{Q}^{n}_{\mathbf{D}} \subseteq \mathcal{P}^{n}_{\mathbf{D}}$ (see Lemma \ref{equivmeas}), 
	$\mathcal{Q}_{\mathbf{D}^{\prime}_n}^{n}  \subseteq \mathcal{Q} $ and $\mathcal{Q}_{\mathbf{D}} \subseteq \mathcal{Q}$. Thus, \eqref{main-thm} holds for all $\xi \in \mathbb{L}^1_{*}$, and hence, holds for all $\xi$ that satisfy condition of Proposition \ref{mthm}.
	
	\textit{First part of \ref{max-eq}}: 
	$\mathcal{A}$ is closed and obviously bounded with respect to the norm $\| \cdot\|_{\infty}$ as $\mathbf{D}^{\prime}_n$ is bounded. By Heine-Borel theorem, $\mathcal{A}$ is a compact subset of a $N(n,n)$-dimensional vector space\footnote{The cardinality of $\mathcal{L}_n$, $\# \mathcal{L}_n = 2n+1$, $ \# \mathcal{L}_n^{n+1} = (2n+1)^{n+1}$, and $\# (\{0, \dots, n\} \times \mathcal{L}_n^{n+1} ) = (n+1)(2n+1)^{n+1} = N(n,n).$
	} equipped with the norm ${\|\cdot\|_{\infty}}$.
	
	\textit{Second part of \ref{max-eq}}: 
	Here, we show that $F: f \mapsto \mathbb{E}^{P_n  \circ (M^{f, \mathbb{X}})^{-1}} [\xi (\widehat{X}^n)]$ is continuous.  
	From Proposition \ref{mthm} we know that $\xi$ is continuous, $\widehat{X}^n$ is the interpolated canonical process, i.e., $\widehat{X}: \mathcal{L}^{n+1}_{n} \rightarrow \Omega,$ thus $\widehat{X}^n$ is continuous and $P_n$ takes it values from the set of real numbers. 
	For $F:f \mapsto \mathbb{E}^{P_n \circ (M^{f, \mathbb{X}})^{-1}} [\xi (\widehat{X}^n)]$ to be continuous, ${\psi: f \mapsto M^{f, \mathbb{X}} }$ has to be continuous. 
	Since $\mathcal{A}$ is a compact subset of a $N(n,n)$-dimensional vector space for fixed $n\in \mathbb{N}$ and ${M^{f, \mathbb{X}}: \Omega_n \rightarrow \mathcal{L}^{n+1}_{n}}$, for all $f,g \in \mathcal{A}$, 
	$$|M^{f, \mathbb{X}}-M^{g, \mathbb{X}}| = | \|f \| _{\infty}-  \|g \| _{\infty} | \leq \| f-g \|_{\infty}.$$ 
	Thus, $\psi$ is continuous with respect to the norm $\| \cdot \|_{\infty}$. Hence $F$ is continuous with respect to any norm on $\mathbb{R}^{N(n,n)}$. 
	
\end{proof}

\begin{thm}
	\label{max-result}
	
	Let $\xi: \Omega \rightarrow \mathbb{R}$ be a continuous function satisfying $|\xi(\omega)| \leq a(1+\| \omega\|_{\infty})^b$ for some constants $a,b >0.$ Then, 
	\begin{equation}
	\label{max-thm}
	\sup_{P \in \mathcal{Q}_{\mathbf{D}}} \mathbb{E}^{P} [\xi] = \lim _{n\rightarrow \infty} \max_{\mathbb{Q} \in \mathcal{Q}_{\mathbf{D}^{\prime}_n/n}^{n}} \mathbb{E}^{\mathbb{Q}} [\xi (\widehat{X}^n)] .
	\end{equation}
\end{thm}
\begin{proof}
	The proof follows directly from Proposition \ref{mthm}.
\end{proof}

\section{Nonstandard construction of $G$-expectation}

\label{hyperfinitecons}

\subsection{Hyperfinite-time setting}
Here we present the nonstandard version of the discrete-time setting of the sublinear expectation and the strong formulation of volatility uncertainty on the hyperfinite timeline.
\begin{dfn}
	
	${}^{*} \Omega$ is the ${}^{*}$-image of $\Omega$ endowed with the ${}^{*}$-extension of the maximum norm $ {}^{*} \|\cdot \| _{\infty}$. 
\end{dfn}
${}^{*} \mathbf{D} ={}^{*}[r_\mathbf{D},R_\mathbf{D}]$ is the ${}^{*}$-image of $\mathbf{D}$, and as such it is {\em internal}. \\
It is important to note that $st : {}^{*} \Omega \rightarrow \Omega$ is the standard part map, and $st(\omega)$ will be referred to as the {\em standard part} of $\omega$, for every $\omega \in  {}^{*} \Omega$. ${}^{\circ} z$ denotes the standard part
of a hyperreal $z$. \\

\begin{dfn}
	\label{neardefinition}
	
	For every $\omega \in \Omega$, if there exists $\widetilde{\omega} \in {}^{*}\Omega$ such that ${\| \widetilde{\omega} - {}^{*}\omega \|_{\infty} \simeq 0}$, then $\widetilde{\omega}$ is a {\em nearstandard point} in ${}^{*} \Omega$. This will be denoted as
	$ns(\widetilde{\omega}) \in  {}^{*}\Omega$.
\end{dfn}

For all hypernatural $N,$ let
\begin{equation}
\mathcal{L} _{N} = \left\lbrace \frac{K}{N\sqrt{N}}, \quad -N^2 \sqrt{R_{\mathbf{D}}} \leq K \leq N^2 \sqrt{R_{\mathbf{D}}}, \quad K \in {}^{*} \mathbb{Z} \right\rbrace,
\end{equation}
and the hyperfinite timelime 
\begin{equation}
\mathbb{T}  = \left\lbrace 0, \frac{T}{N}, \cdots, -\frac{T}{N}+T, T  \right\rbrace.
\end{equation}
We consider $\mathcal{L}^{\mathbb{T}} _{N}$ as the canonical space of paths on the hyperfinite timeline, and $X^N = {(X^N_k)}_{k=0}^{N}$ as
the canonical process denoted by $X^N_k(\bar{\omega})= \bar{\omega}_k$ for ${\bar{\omega} \in \mathcal{L}^{\mathbb{T}} _{N}}$. $\mathcal{F}^N$ is the internal filtration generated by $X^N$.
The linear interpolation operator can be written as 
$$\text{  }\widetilde{  } \text{ }: \text{ }\widehat{\cdot} \text{ } \circ \text{  } \iota ^{-1} \rightarrow {}^{*} \Omega , \quad \text{for  } \widetilde{\mathcal{L}^{\mathbb{T}} _{N}}\subseteq {}^{*} \Omega,$$ 
where $$\widehat{\omega}(t) := (\lfloor Nt/T \rfloor +1 -Nt/T)\omega_{ \lfloor Nt/T \rfloor } + (Nt/T - \lfloor Nt/T \rfloor)\omega_{\lfloor Nt/T \rfloor +1},$$
for $\omega \in \mathcal{L}^{N+1} _{N}$ and for all $t \in {}^*[0,T]$. $\lfloor y \rfloor$ denotes the greatest integer less than or equal to $y$ and $\iota : \mathbb{T} \rightarrow \{0, \cdots, N\}$ for $\iota : t \mapsto Nt/T$.

For the hyperfinite strong formulation of the volatility uncertainty, fix $N \in { }^{*}\mathbb{N} \setminus \mathbb{N}$. Consider $\left\lbrace \pm \frac{1}{\sqrt{N}} \right\rbrace ^{\mathbb{T}},$ and let $P_N$
be the uniform counting measure on $\left\lbrace \pm \frac{1}{\sqrt{N}} \right\rbrace ^{\mathbb{T}}$. $P_N$ can also be seen as a measure on $\mathcal{L}^{\mathbb{T}} _{N}$, concentrated on $\left\lbrace \pm \frac{1}{\sqrt{N}}\right\rbrace ^{\mathbb{T}}$.
Let ${\Omega _{N} = \{ \underline{\omega} = (\underline{\omega}_1, \cdots, \underline{\omega}_N); \underline{\omega}_i = \{\pm 1\}, i =1, \cdots
	, N\}}$, and let $\Xi_1, \cdots, \Xi_N$ be a ${}^{*}$-independent sequence of $\{\pm 1 \}$-valued random variables on $\Omega _{N}$
and the components of $\Xi _k$ are orthonormal in $L^2(P_N)$. We denote the hyperfinite random walk by
$$\mathbb{X}_t = \frac{1}{\sqrt{N}} \sum_{l=1}^{ Nt/T} \Xi _l \quad \text{ for all } t \in \mathbb{T}  .$$
The hyperfinite-time stochastic integral of some $F:\mathbb{T} \times \mathcal{L}^{\mathbb{T}} _{N} \rightarrow {}^* \mathbb{R}$ with respect to the hyperfinite random walk is given by
$$\sum_{s=0}^{t} F(s,\mathbb{X}) \Delta \mathbb{X}_s: \Omega _{N} \rightarrow {}^{*} \mathbb{R}, \quad \underline{\omega} \in \Omega _{N} \mapsto \sum_{s=0}^{t} F(s,\mathbb{X}(\underline{\omega})) \Delta \mathbb{X}_s(\underline{\omega}).$$
Thus, the hyperfinite set of martingale laws can be defined by
\begin{align}
\label{disc-meas} 
\bar{\mathcal{Q}}_{\mathbf{D}^{\prime}_N}^{N} = \left\{ \begin{array}{lr}   P_N  \circ ( M^{F,\mathbb{X} })^{-1}; \text{  } F: \mathbb{T} \times \mathcal{L}^{\mathbb{T}}_{N} \rightarrow \sqrt{\mathbf{D}^{\prime}_{N}} \nonumber
\end{array}  \right\}
\end{align}
where
$${\mathbf{D}^{\prime}_N} = {}^{*}\mathbf{D} \cap \left( \frac{1}{N} {}^*\mathbb{N} \right)^2$$
and 
$$M^{F, \mathbb{X}} = \left( \sum _{s=0}^{t} F (s, \mathbb{X})\Delta \mathbb{X}_ s \right) _{t\in \mathbb{T}}.$$

\begin{rem}
	Up to scaling,  $\bar{\mathcal{Q}}_{\mathbf{D}^{\prime}_N}^{N} = \mathcal{Q}_{\mathbf{D}^{\prime}_n}^n$.
\end{rem}

\subsection{Results and proofs} 
\label{resultchap}

\begin{dfn}[(Uniform lifting of $\xi$)]
	\label{dfn-lifting}
	
	Let $\Xi: \mathcal{L}^{\mathbb{T}} _{N} \rightarrow {}^{*} \mathbb{R}$ be an internal function, and let $\xi: \Omega \rightarrow \mathbb{R}$ be a continuous function.  $\Xi$ is said to be a {\em uniform lifting} of $\xi$ if and only if
	$$\forall\bar{\omega} \in  \mathcal{L}^{\mathbb{T}} _{N} \Big(\widetilde{\bar{\omega}} \in ns({}^{*} \Omega) \Rightarrow {}^{\circ} \Xi (\bar{\omega}) = \xi (st(\widetilde{\bar{\omega}} )) \Big),$$
	where $st(\widetilde{\bar{\omega}})$ is defined 
	with respect to the topology of uniform convergence on $\Omega$.
\end{dfn}

In order to construct the hyperfinite version of the $G$-expectation, we need to show that the ${ }^*$-image of $\xi$, ${}^{*}\xi$, with respect to $\widetilde{\bar{\omega}} \in ns( {}^{*} \Omega)$, is the canonical lifting of $\xi$
with respect to $st(\widetilde{\bar{\omega}}) \in \Omega$. i.e., for every $\widetilde{\bar{\omega}} \in ns( {}^{*} \Omega)$, ${}^{\circ} \left({}^{*}\xi (\widetilde{\bar{\omega}})\right)  = \xi (st(\widetilde{\bar{\omega}}))$. To do this, we need to show that ${ }^{*} \xi$ is S-continuous in every nearstandard point $\widetilde{\bar{\omega}}$. 

It is easy to prove that there are two equivalent characteristics of $S$-continuity on ${}^* \Omega$.
\begin{rem}
	\label{s-contr}
	The following are equivalent for an internal function ${\Phi: { }^{*}\Omega \rightarrow { }^{*} \mathbb{R}}$;
	\begin{itemize}
		\label{s-cont}
		\item[$(1)$] $\forall \omega^{'} \in { }^{*}\Omega \left({ }^{*} \| \omega - \omega^{'} \|_{\infty} \simeq 0 \Rightarrow  { }^{*} |\Phi (\omega)- \Phi (\omega ^{'})|\simeq 0\right).$
		\item[$(2)$] $\forall \varepsilon \gg 0, \exists \delta \gg 0: \forall \omega^{'} \in { }^{*}\Omega \left( { }^{*} \| \omega - \omega^{'} \|_{\infty} < \delta \Rightarrow { }^{*} |\Phi (\omega)-  \Phi (\omega ^{'})| < \varepsilon \right).
		$
	\end{itemize}
\end{rem}
(The case of Remark \ref{s-cont} where $\Omega = \mathbb{R}$ is well known and proved in \citet[Theorem $5.1.1$]{Stroyan})

\begin{dfn}
	\label{s-contd}
	
	Let ${\Phi: { }^{*}\Omega \rightarrow { }^{*} \mathbb{R}}$ be an internal function. We say $\Phi$ is {\em$S$-continuous} in $\omega \in { }^{*}\Omega$, if and only if it satisfies one of the two equivalent conditions of Remark \ref{s-contr}. 
\end{dfn}


\begin{pro}
	\label{s-pro}
	
	If ${\xi: \Omega \rightarrow \mathbb{R}}$ is a continuous function satisfying $|\xi (\omega)| \leq a(1+ \|\omega \|_{\infty})^b$, for $a,b>0$, then, $\Xi= {}^{*}\xi \circ \widetilde{\cdot}$ is a uniform lifting of $\xi$. 
\end{pro}

\begin{proof}
	Fix $\omega \in \Omega$. By definition, $\xi$ is continuous on $\Omega$. i.e., for all $\omega \in \Omega$, and for every $ \varepsilon \gg 0$, there is a $\delta \gg 0 $, such that for every $ \omega^{'} \in \Omega$, if 
	\begin{equation}
	\label{xi-cont}
	\| \omega - \omega^{'} \|_{\infty} <  \delta , \text{ then } |\xi (\omega) - \xi (\omega^{'})| <  \varepsilon. 
	\end{equation}
	By the Transfer Principle: For all $\omega \in \Omega$, and for every $ \varepsilon \gg 0$, there is a $\delta \gg 0 $, such that for every $ \omega^{'} \in {}^{*} \Omega$, \eqref{xi-cont} becomes, 
	\begin{equation}
	\label{s-del}
	{}^{*} \|{}^{*} \omega - \omega^{'} \|_{\infty} <  \delta , \text{ and } {}^{*} | {}^{*} \xi ({}^{*}\omega) - {}^{*}  \xi (\omega^{'})| <  \varepsilon. 
	\end{equation}
	So, ${}^{*} \xi$ is $S$-continuous in ${}^{*} \omega$ for all $\omega \in \Omega$. Applying the equivalent characterization of $S$-continuity, Remark \ref{s-contr}, \eqref{s-del} can be written as \
	\begin{equation*}
	{ }^{*} \| { }^{*} \omega - \omega^{'} \|_{\infty} \simeq 0, \text{ and } { }^{*} |{ }^{*} \xi ({ }^{*} \omega) - { }^{*} \xi (\omega^{'})| \simeq 0.
	\end{equation*}
	We assume $\widetilde{\bar{\omega}}$ to be a nearstandard point. By Definition \ref{neardefinition}, this simply implies, 
	\begin{equation}
	\label{near}
	\forall \widetilde{\bar{\omega}} \in ns( {}^{*} \Omega), \text{  } \exists \omega \in \Omega :  { }^{*} \| \widetilde{\bar{\omega}} -  { }^{*}\omega  \|_{\infty} \simeq 0.
	\end{equation}
	Thus, by $S$-continuity of ${ }^{*} \xi$ in ${}^{*} \omega$,\\
	$${ }^{*} |{ }^{*} \xi (\widetilde{\bar{\omega}})    - { }^{*} \xi ({ }^{*} \omega)| \simeq 0. $$
	Using the triangle inequality, if $\omega^{'}  \in {}^{*} \Omega$ with ${}^*\|   \widetilde{\bar{\omega}}  - \omega^{'} \|_{\infty}\simeq 0 $,
	$${ }^{*} \|{ }^{*}\omega  - \omega^{'} \|_{\infty} \leq   { }^{*} \|  { }^{*}\omega - \widetilde{\bar{\omega}} \|_{\infty} + { }^{*}\| \widetilde{\bar{\omega}} - \omega^{'} \|_{\infty} \simeq 0$$
	and therefore again by the $S$-continuity of ${ }^{*} \xi$ in $ { }^{*} \omega$,
	$${ }^{*} |{ }^{*} \xi ({}^{*} \omega)    - { }^{*} \xi (\omega ^{'})|  \simeq 0 .$$
	And so,
	\begin{align*}
	{ }^{*} |{ }^{*} \xi (\widetilde{\bar{\omega}})- { }^{*} \xi (\omega ^{'})| & \leq { }^{*} |{ }^{*} \xi (\widetilde{\bar{\omega}}) - { }^{*} \xi ({ }^{*} \omega)| + { }^{*} |{ }^{*} \xi ({ }^{*}\omega)- { }^{*} \xi (\omega^{'})|
	\simeq 0. &
	\end{align*}
	Thus, for all $ \widetilde{\bar{\omega}} \in ns( {}^{*} \Omega)$ and $ \omega^{'} \in {}^{*} \Omega  $, if ${ }^{*} \| \widetilde{\bar{\omega}} - \omega^{'} \|_{\infty} \simeq 0$, then,
	\begin{equation*}
	{ }^{*} |{ }^{*} \xi (\widetilde{\bar{\omega}}) - { }^{*} \xi (\omega ^{'})|\simeq 0. 
	\end{equation*}
	Hence, ${ }^{*} \xi$ is S-continuous in $\widetilde{\bar{\omega}}$. Equation \eqref{near} also implies
	$$\widetilde{\bar{\omega}} \in m (\omega) \left( m (\omega) = \bigcap \{{ }^{*} \mathcal{O}; \mathcal{O} \text{ is an open neighbourhood of }\omega \}\right)$$
	such that $\omega$ is unique, and in this case $st(\widetilde{\bar{\omega}} ) = \omega$.\\
	Therefore, 
	\begin{equation*}
	{}^{\circ} \Big({}^{*}\xi (\widetilde{\bar{\omega}})\Big)  = \xi (st(\widetilde{\bar{\omega}})).
	\end{equation*}
\end{proof}

\begin{dfn}
	
	Let $\bar{\mathcal{E}}: {{ }^{*} \mathbb{R}}^{\mathcal{L}_{N}^{\mathbb{T}}} \rightarrow { }^{*} \mathbb{R}$. We say that $\bar{\mathcal{E}}$ lifts $\mathcal{E}^G$ if and only if for every $\xi:\Omega \rightarrow \mathbb{R}$ that satisfies $|\xi (\omega)| \leq a(1+ \|\omega \|_{\infty})^b$ for some $a, b> 0$,
	$$\bar{\mathcal{E}}({ }^{*}\xi \circ \tilde{\cdot})  \simeq  \mathcal{E}^G (\xi). $$ 
\end{dfn}

\begin{thm}
	\begin{equation}
	\label{G-thm}
	\max _{\bar{Q}\in \bar{\mathcal{Q}}^N_{\mathbf{D}^{\prime}_N}} \mathbb{E}^{\bar{Q}}[\cdot] \text{  lifts  }\mathcal{E}^G(\xi) .
	\end{equation}
\end{thm}

\begin{proof}
	From Theorem \ref{max-result}, 
	\begin{equation}
	\label{stand-G}
	\max_{\mathbb{Q} \in \mathcal{Q}_{\mathbf{D}^{\prime}_n}^{n}} \mathbb{E}^{\mathbb{Q}} [\xi (\widehat{X}^n)] \rightarrow \mathcal{E}^G(\xi), \quad \text{as } n\rightarrow \infty  .
	\end{equation}
	For all $N \in { }^{*}\mathbb{N} \setminus \mathbb{N},$ we know that \eqref{stand-G} holds if and only if 
	\begin{equation}
	\label{stand-G2}
	\max _{Q\in { }^{*}\mathcal{Q}^N_{\mathbf{D}^{\prime}_N}} \mathbb{E}^Q [{ }^{*} \xi (\widehat{X}^N)]\simeq \mathcal{E}^G (\xi), 
	\end{equation}
	(see \citet{Albeverio}, Proposition $1.3.1$). Now, we want to express \eqref{stand-G2} in term of $\bar{\mathcal{Q}}_{\mathbf{\mathbf{D}^{\prime}_N}}^{N}$. i.e., to show that
	\begin{equation*}
	\max _{ \bar{Q} \in \bar{\mathcal{Q}}^N_{\mathbf{D}^{\prime}_N}} \mathbb{E}^{\bar{Q}}[{ }^{*}\xi \circ \tilde{\cdot}]\simeq \mathcal{E}^G(\xi) .
	\end{equation*}
	To do this, use
	$$\mathbb{E}^Q [{ }^{*} \xi \circ \hat{ \cdot}] = \mathbb{E}^Q [{ }^{*} \xi \circ \hat{\cdot} \circ \iota ^{-1} \circ \iota]$$
	and
	\begin{align*}
	\label{tran-measure}
	\mathbb{E}^Q [{ }^{*} \xi \circ \hat{\cdot} \circ \iota ^{-1} \circ \iota]& = \mathbb{E}^Q [{ }^{*} \xi \circ \tilde{\cdot} \circ \iota] \\ \nonumber
	& = \int _{{ }^{*}\mathbb{R}^{N+1}}{ }^{*} \xi \circ \tilde{\cdot} \circ \iota dQ, \quad \text{(transforming measure)}\\ \nonumber
	& = \int _{{ }^{*}\mathbb{R}^{\mathbb{T}}}{ }^{*} \xi \circ \tilde{\cdot} d(Q \circ j), \\ \nonumber 
	& = \mathbb{E}^{Q \circ j} [{ }^{*} \xi \circ \tilde{\cdot}]
	\end{align*}
	for $j : { }^{*}\mathbb{R}^{\mathbb{T}} \rightarrow { }^{*}\mathbb{R}^{N+1}$, $(xt)_{t \in \mathbb{T}} \mapsto\left( \frac{xNt}{T} \right)_{t \in \mathbb{R}^{N+1}}.$\\
	Thus, $$\bar{\mathcal{Q}}^N_{\mathbf{D}^{\prime}_N} = \{ Q \circ j : Q \in { }^{*}\mathcal{Q}^N_{\mathbf{D}^{\prime}_N} \} .$$
	This implies, 
	$$\max _{ \bar{Q} \in \bar{\mathcal{Q}}^N_{\mathbf{D}^{\prime}_N}} \mathbb{E}^{\bar{Q}}[{ }^{*}\xi \circ \tilde{\cdot}] = \max _{Q\in { }^{*}\mathcal{Q}^N_{\mathbf{D}^{\prime}_N}} \mathbb{E}^Q [{ }^{*} \xi \circ \hat{\cdot}].$$
\end{proof}


\section*{Appendix}

\begin{proof}[Proof of Lemma \ref{equivmeas}]
	From the above equation, we can say that $\Delta M^f_{k} = f(k,\mathbb{X})\xi_k$. And by the orthonormality property of $\xi_k$, we have 
	$$\mathbb{E}^{P_n}[f(k,\mathbb{X})^2\xi_k^2|\mathcal{F}^n_{k}] = \mathbb{E}^{P_n}[f(k,\mathbb{X})^2| \mathcal{F}^n_{k}] \leq \mathbb{E}^{P_n}[(\sqrt{R_{\mathbf{D}}})^2|\mathcal{F}^n_{k}] = R_{\mathbf{D}} \quad P_n\text{ a.s.},$$
	as $|\xi_k|=1$, $f(\cdots)^2\in \mathbf{D}$ implies 
	$$| (\Delta M_{k}^f)^2 | = |f(k,\mathbb{X})|^2 \in [r_{\mathbf{D}}, R_{\mathbf{D}}]\quad P_n\text{ a.s}.$$
\end{proof}

\subsection*{Density argument verification}
Let
$$f: \xi \mapsto \sup_{P \in \mathcal{Q}_\mathbf{D}} \mathbb{E}^{P} [\xi] $$  
and 
$$g: \xi \mapsto  \lim_{n \rightarrow \infty} \sup_{\mathbb{Q}\in \mathcal{Q}^{n}_{{\mathbf{D}^{\prime}_n}/n}} \mathbb{E}^{\mathbb{Q}} [\xi (\widehat{X}^n)].$$
From \eqref{final-equation}, we know that for all $\xi \in \mathcal{C}_{b}(\Omega, \mathbb{R}), \text{ } f(\xi) = g(\xi).$  
Since $\mathbb{L}^{1}_{*}$ is the completion of $\mathcal{C}_{b}(\Omega, \mathbb{R})$ under the norm $\| \cdot \|_{*}$,  $\mathcal{C}_{b}(\Omega, \mathbb{R})$
is dense in $\mathbb{L}^{1}_{*};$ and we want to prove for all $\xi \in \mathbb{L}^{1}_{*}$, $f(\xi) = g(\xi).$ To prove this, it is sufficient to show that $f$ and $g$ 
are continuous with respect to the norm $\| \cdot \| _{*}.$

\subsubsection*{For continuity of $f$:}

For all $ P \in \mathcal{Q}_\mathbf{D}$ and $\xi, \xi^{'} \in \mathbb{L}^{1}_{*}$, 

$$\sup _{ P \in \mathcal{Q}_\mathbf{D}}\mathbb{E}^{P} [\xi] - \sup_{ P \in \mathcal{Q}_\mathbf{D}} \mathbb{E}^{P} [\xi ^{'}] \leq \sup _{ P \in \mathcal{Q}_\mathbf{D}}\mathbb{E}^{P} [| \xi - \xi^{'} |].$$
Since, $\mathcal{Q}_\mathbf{D} \subseteq \mathcal{Q}$,
\begin{equation}
\label{ap-conf}
\sup _{ P \in \mathcal{Q}_\mathbf{D}}\mathbb{E}^{P} [\xi] - \sup_{ P \in \mathcal{Q}_\mathbf{D}} \mathbb{E}^{P} [\xi ^{'}] \leq \| \xi - \xi ^{'} \| _{*}. 
\end{equation}
Interchanging $\xi$ and $\xi^{'}$,

\begin{equation}
\label{ap-conf1}
\sup _{ P \in \mathcal{Q}_\mathbf{D}}\mathbb{E}^{P} [\xi ^{'}] - \sup_{ P \in \mathcal{Q}_\mathbf{D}} \mathbb{E}^{P} [\xi] \leq \| \xi ^{'} - \xi  \| _{*}. 
\end{equation}
Adding \eqref{ap-conf} and \eqref{ap-conf1}, we have 
$|f(\xi) - f(\xi ^{'})  | \leq  \| \xi  - \xi ^{'} \| _{*}. $
\subsubsection*{For continuity of $g$:}
We follow the same argument as above.

\begin{proof}[Proof of Remark \ref{s-contr}]
	Let $\Phi$ be an internal function such that condition $(1)$ holds. To show that $(1) \Rightarrow (2)$, fix $\varepsilon \gg 0$. We shall show there exists a $\delta$ for 
	this $\varepsilon$ as in condition $(2)$. Since $\Phi$ is internal, the set
	$$I = \left\lbrace \delta \in  { }^{*} \mathbb{R}_{>0}: \text{ }\forall \omega^{'} \in { }^{*}\Omega \text{ } ( { }^{*} \| \omega - \omega^{'} \|_{\infty} < \delta \Rightarrow { }^{*} |\Phi (\omega)-  \Phi (\omega ^{'})| < \varepsilon ) \right\rbrace,$$
	is internal by the Internal Definition Principle and also contains every positive infinitesimal. By Overspill (cf. \citet[Proposition $1.27$]{Albeverio}) $I$ must then contain some positive $\delta \in \mathbb{R}$.\\
	Conversely, suppose condition (1) does not hold, that is, there exists some $\omega^{'} \in { }^{*}\Omega$ such that   
	$${ }^{*} \| \omega - \omega^{'} \|_{\infty} \simeq 0 \text{ and }  { }^{*} |\Phi (\omega)- \Phi (\omega ^{'})| \text{ is not infinitesimal}.$$
	If $\varepsilon = \min(1, { }^{*}| \Phi (\omega)- \Phi (\omega ^{'})|/2),$ we know that for each standard $\delta >0,$ there is a point $\omega ^{'}$ within $\delta$ of $\omega$ at which $\Phi (\omega ^{'})$ is farther than $\varepsilon$ 
	from $\Phi (\omega)$. This shows that condition $(2)$ cannot hold either. 
\end{proof}

\end{document}